\documentclass{article}
\usepackage{ijcai22}

\usepackage{times}

\usepackage{soul}
\usepackage{url}
\usepackage{hyperref}
\usepackage[utf8]{inputenc}
\usepackage[small]{caption}
\usepackage{graphicx}
\usepackage{amsmath}
\usepackage{booktabs}
\usepackage{algorithm}
\usepackage[noend]{algpseudocode}
\usepackage{amssymb}
\usepackage{amsfonts}
\usepackage{mathtools}
\usepackage{tikz}
\usetikzlibrary{calc}
\usepackage{todonotes}
\usepackage{dsfont}
\usepackage{enumitem}
\usepackage{tikz}
\usepackage{amsthm}
\usepackage{thm-restate}
\usepackage[capitalise,noabbrev]{cleveref}
\usepackage{multirow}
\usepackage{newfloat}
\usepackage{wrapfig}
\usepackage{dsfont}
\usepackage{comment}
\usepackage{mleftright}
\usepackage{mdframed}
\usepackage{bm}
\usepackage{multirow}
\usepackage[mathscr]{eucal}
\usepackage{subfig}

\newcommand{\I}{I}
\newcommand{\X}{\mathcal{X}}

\newcommand{\V}{\mathcal{V}}
\newcommand{\B}{\mathcal{B}}
\newcommand{\gr}{\mathcal{G}}

\newcommand{\cG}{\mathscr{G}}
\newcommand{\cF}{\mathscr{F}}
\newcommand{\expe}{\mathbb{E}}

\newcommand{\bb}[1]{\bm{#1}}
\newcommand{\one}{\mathds{1}}

\newcommand{\bv}{\bb{b}}
\newcommand{\vv}{\bb{v}}
\newcommand{\vvec}{\bb{v}}

\newcommand{\poa}{\textsc{PoC}}
\newcommand{\gammav}{\bb{\gamma}}

\newcommand{\bcce}{\Sigma^{\ast}}
\newcommand{\alg}{\textsc{Alg}}
\newcommand{\defeq}{\coloneqq}%

\newcommand\add[1]{\textcolor{black}{#1}}

\newtheorem{assumption}{Assumption}

\newtheorem{remark}{Remark}

\newtheorem{definition}{Definition}

\title{Fair Equilibria in Sponsored Search Auctions: The Advertisers’ Perspective}

\author{
 Georgios Birmpas$^1$\and
 Andrea Celli$^2$\and
 Riccardo Colini-Baldeschi$^{3}$\and
 Stefano Leonardi$^1$\\
 \affiliations
 $^1$Department of Computer, Control and Management Engineering, Sapienza University, Rome, Italy\\
 $^2$Department of Computing Sciences, Bocconi University, Milan, Italy\\
 $^3$Core Data Science, Meta, London, UK\\
 \emails
 \{birbas, leonardi\}@diag.uniroma1.it,
andrea.celli2@unibocconi.it,
rickuz@fb.com
}

\begin{document}

\maketitle

\begin{abstract}
    In this work we introduce a new class of mechanisms composed of a traditional Generalized Second Price (GSP) auction and a fair division scheme, in order to achieve some desired level of fairness between groups of Bayesian strategic advertisers.
    We propose two mechanisms, $\beta$-Fair GSP and GSP-EFX, that compose GSP with, respectively, an \emph{envy-free up to one item}, and an \emph{envy-free up to any item} fair division scheme. The payments of GSP are adjusted in order to compensate advertisers that suffer a loss of efficiency due the fair division stage. 
    We investigate the strategic learning implications of the deployment of sponsored search auction mechanisms that obey to such fairness criteria. We prove that, for both mechanisms, if bidders play so as to minimize their external regret they are guaranteed to reach an equilibrium with good social welfare. We also prove that the mechanisms are budget balanced, so that the payments charged by the traditional GSP mechanism are a good proxy of the total compensation offered to the advertisers. Finally, we evaluate the quality of the allocations through experiments on real-world data. 
\end{abstract}

\section{Introduction}
    
Over the last decades, online advertising has been one of the main tools for small and medium business (SMB) to grow. Online advertising allows SMBs to reach potential customers without geographical or demographic barriers.
Moreover, it offers better return-on-investments than other advertising mediums thanks to its highly personalized system.
Given its crucial role in the growth of businesses (and, in its turn, society), it is natural to study how the mechanisms implemented in online advertising platforms can be improved to obey to different fairness criteria for advertisers.
There are indeed various settings where one may care about balancing ads allocations, even just between a majority and a minority group of advertisers. 
For example, one setting is when large companies and small businesses are competing for the same set of users, and the platform may want to ensure that small companies get reasonable visibility despite smaller budgets. 
As another example, there may be businesses based in different geographical locations but offering similar products/services. These businesses may have unbalanced budgets due to their location, but one may want to guarantee to each business some visibility in a target set of users regardless of geographical attributes. More in general, one could think of incorporating fairness in sponsored search auctions whenever two different parties are advertising to increase users’ awareness about some sensitive topic.

We study how fairness notions borrowed from the fair division literature can be used to model fairness with respect to advertisers. 
Notions from fair division already found applications in online advertising~\cite{chawla2020fairness,ilvento2020multi}. In the fair division literature, the dominant notion of fairness aims at providing guarantees for individual agents (see, e.g.,~\cite{brandt2016handbook,moulin2003fair}). However, we argue that a \emph{group approach} would be more practical, better aligned with societal expectations, and easier to implement.
Following recent works in group fair division literature~\cite{kyropoulou2020almost,manurangsi2017asymptotic,manurangsi2019computing,segal2019democratic}, we propose to use \emph{envy-freeness}~\cite{foley1967} to study \emph{group fairness}. \emph{Group envy-freeness} guarantees that no group of agents envies the allocation obtained by any other group. Unfortunately, envy-freeness cannot be guaranteed for indivisible items even in simple settings with two agents and one item. 
Thus, we focus on two natural relaxations: envy-freeness up to \emph{one} good (EF1)~\cite{Budish11,lipton2004approximately}, and envy-freeness up to \emph{any} good (EFX)~\cite{CaragiannisKMPS19,GMT14}.

In practice, any attempt to guarantee such properties in real advertising platforms will inevitably collide with real-world engineering constraints. Therefore, a credible solution should be a mechanism that can be easily integrated with a  pre-existing auction framework, without requiring substantial changes to it. 
We focus on the \emph{generalised second price} (GSP) auction framework~\cite{edelman2007internet}, which is one of the most frequently adopted mechanisms for the allocation of advertising opportunities in large Internet advertising companies.  
In this setting, we show the existence of simple mechanisms that guarantee
some notion of group EF1 (resp.,  group EFX) for advertisers. Such mechanisms can be implemented as a post-auction layer to be run after a standard GSP mechanism. 
In the spirit of the work by~\cite{DworkI19}, we study the properties of such composite mechanisms. 

\paragraph{Original contributions}
We focus on a Bayesian setting with incomplete information (i.e., the valuations for advertising opportunities are stochastic, and each bidder does not observe the realized valuations of the other bidders). For each auction, bidders are divided in two groups (a majority group and a minority group) based on their characteristics and competitiveness. Given the different characteristics, the users interact in different ways with the ads from the two groups. This is modeled through group specifics click-through rates and quality factors.\footnote{The online advertising problem with group specifics click-through rates has already been formalized in literature in the Ad-Types setting \cite{Colini-Baldeschi20,elzayn2021equilibria}.}
First, in \cref{sec:fair gsp auctions}, we introduce two notions of group envy-freeness (\emph{group $\beta$}-EF1 and \emph{group $\beta$}-EFX), parameterized on a factor $\beta$ which allows the platform to tune the strength of the fairness requirement. We show that group $\beta$-EF1 and group $\beta$-EFX allocations always exist for two groups with monotonic click-through rates. Moreover, they can be computed efficiently with two fair division schemes which are, respectively,  ``group versions'' of the  round robin procedure, and the envy-cycle elimination algorithm by Lipton \emph{et al.}~\shortcite{lipton2004approximately}. 
We define fairness with respect to reported bidder's valuations. %
Indeed, we argue that an auction is perceived to be fair if the  slots are divided between the majority/minority group in a fair way, given that prices are not publicly known to all advertisers in real-world sponsored search auctions.  Considering the social welfare of the allocation is also the customary approach in the fair division literature, in which notions such as EF1 and EFX are  defined on valuations (see, e.g.,  \cite{Budish11,DBLP:journals/corr/abs-1909-07650}).
Then, in \cref{sec:efficiency and bb}, we study the efficiency and budget-balance of the mechanisms resulting from the composition of GSP with the two fair division schemes. We show that the social welfare of the composite mechanisms is a good approximation of the optimum. Moreover, we prove that the welfare loss experienced by the advertisers due to the fairness constraints can be partly compensated through monetary incentives via redistribution of the GSP payments.
Finally, \cref{sec:poc} studies the behavior of the composite mechanisms when bidders behave as no-regret learning agents (i.e., they take decisions so as to minimize their external regret) and they are conservative (i.e.,  they do not overbid).\footnote{The assumption of having conservative bidders has often been used in the analysis of ad auction mechanisms (see, e.g.,~\cite{PT10,CKS16,Feldman2013SimultaneousAA}).  The rationale for this is usually that, in practice, bidders are only partially informed and, thus, they may be more inclined to avoid risks (i.e., to submit bids that might yield negative payoffs). } 
We show that the learning dynamic originating from the interaction with our composite mechanisms converges to a \emph{Bayesian coarse correlated equilibrium} with good social welfare properties characterized through their {\em price of composition}.
We complement our theoretical results by evaluating the quality of our fair sponsored search mechanisms on real-world data (\cref{sec:exp}).

\section{Preliminaries}

    Throughout the paper, bold case letters denote column vectors. Given a vector $\bb{y}$, its $i$-th component is denoted by $y_i$. 
    The set $\{1,\ldots,x\}$ is denoted by $[x]$, and $\Delta_\X$ is the $|\X|$-dimensional simplex over the discrete set $\X$.
    
    \paragraph{Sponsored search framework}
    There is a set $\I$ of $n$ bidders and a set $J$ of $m$ slots. An \emph{outcome} is an assignment of bidders to slots. Each bidder $i$ has a private type $v_i$, representing their valuation on the item which is being sold. The vector of types is denoted as $\bb{v}=(v_1,\ldots,v_n)$. 
    Each bidder $i$ belongs to a group from a finite set of possible groups $\gr$. The function $g:[n]\to \gr$ maps bidders to their group, that is, we write $g(i)$ to denote the group to which bidder $i$ belongs.
    We assume that bidders may belong to one of two groups $\gr=\{h,\ell\}$ (e.g., a majority group and a minority group) and the two groups may have different sizes.
    We denote the set of bidders belonging to the two groups by $I_h$ and by $I_\ell$. 
    \add{We make the assumption that $|I_h|\ge m$ and $|I_\ell|\ge m$. This assumption is reasonable in the context of large Internet advertising markets.}
    As it is customary in the literature, we use the model of separable click probabilities (see, e.g.,~\cite{edelman2007internet,varian2007position}), in which each slot $j$ is associated with a \emph{click-through rate} $\alpha_{j,g(i)}$ for group $g(i)$.
    We assume that, for each group $w\in \gr$, $\alpha_{1,w}\ge\alpha_{2,w}\ge\ldots\ge\alpha_{m,w}$, and without loss of generality we take $n=m$. Group-specific click-through rates can model, for example, cases in which different groups of advertisers mainly resort to different media types.
    Finally, each group $w\in\gr$ is associated with a \emph{quality factor} $\gamma_w\in [0,1]$, which reflects the \emph{clickability} of ads from bidders belonging to group $w$.
    For example, ads coming from advertisers in the majority group may have higher clickability than ads from advertisers in the minority groups because of differences in the available budgets to develop the campaigns.
    Quality factors are private knowledge of the advertising platform, and not known by the bidders.~\footnote{
        Our results hold also in the case of advertiser-dependent clickability, because that does not alter the relative ordering of items. 
    }
    
   \paragraph{GSP auction} A \emph{mechanism} elicits a bid $b_i\in\mathbb{R}_{\ge 0}$ for each bidder $i$, which is interpreted as a type declaration, and computes an outcome as well as a price $p_i(\bv,\bb{\gamma})$ for each bidder $i$. 
    We denote by $\pi(\bv,\bb{\gamma},j)$ the bidder assigned to slot $j$ when the mechanism observes the bid vector $\bv$ and vector of quality factors $\bb{\gamma}$.  We also denote by $\nu(\bv,\bb{\gamma},i)$ the slot assigned to bidder $i$ when the mechanism observes the bid vector $\bv$ and vector of quality factors $\bb{\gamma}$.
    When the vectors of bids and quality factors are clear from the context we simplify the notation by writing $\pi(j)$, $\nu(i)$, and $p_i$ 
    in place of $\pi(\bv,\bb{\gamma},j)$, $\nu(\bv,\bb{\gamma},i)$, and $p_i(\bv,\bb{\gamma})$, respectively.
    The value perceived by bidder $i$ when they are allocated $\nu(i)$ is $\alpha_{\nu(i), g(i)}\gamma_{g(i)} v_i$, and their utility is
    $u_i(\bv,\vvec,\gammav)\defeq\alpha_{\nu(i),g(i)} \gamma_{g(i)}v_i-p_i$.
    We focus on a family of mechanisms derived from the \emph{Generalized Second Price} (GSP) auction~\cite{varian2007position}.
    In a GSP auction the mechanism assigns the slots in order from $1$ to $m$ and sets $\pi(\bv,\bb{\gamma},j)$ to be the bidder with the highest \emph{effective bid} $\gamma_{g(i)}\alpha_{j,g(i)} b_i$ not yet assigned (breaking ties arbitrarily). 
    For any bid profile $\bv$, quality factors $\gammav$ and for each $j\in[m]$, $i=\pi(j)$, the price charged to bidder $i$ is computed as
    \begin{equation}\label{eq:prices}
    p^\textsc{g}_i(\bv,\gammav)\defeq\frac{\gamma_{g(\pi(j+1))} \alpha_{j,g(\pi(j+1))}}{\gamma_{g(i)}} b_{\pi(j+1)},
    \end{equation}
    where we set $b_{n+1}=0$.\footnote{Alternatively, we could charge bidder $\pi(j)$ with the \emph{threshold price}, which is the smallest effective bid $p_i^\textsc{t}$ that guarantees them the same slot. Observe that $p_i\leq p_i^\textsc{t}$ since the next highest bidder on $j$ is not necessarily $\pi(j+1)$ if the bidders have different quality factors.}
    The mechanism is {\em Individually Rational} (IR) if, for each bidder $i\in \I$, $u_i(\bv,\vvec,\gammav)\geq 0$, for all $\bv$, $\vvec$, and $\gammav$.   The mechanism is  {\em Individually Rational at the Equilibrium} (IRE) if it is IR  at the equilibrium bid vectors.

    \paragraph{Online Bayesian framework} The $n$ bidders participate in a series of GSP auctions. 
    At each iteration $t$, each bidder $i$ observes a valuation $v^t_i$ for the item being sold at time $t$. Let $\V_i$ be the finite set of types of bidder $i$.
    The vector of types $\vvec^t=(v_1^t,\ldots,v_n^t)$ is drawn, at each $t$, from a (possibly correlated) probability distribution $\cF$ supported on a finite set of joint types $\V$, that is, $\V\defeq\times_{i\in[n]}\V_i$.
    Moreover, at each $t$, a vector of quality factors $\gammav^t\in [0,1]^{|\gr|}$ is drawn from a (possibly correlated) distribution $\cG$.
    Each bidder $i$ has an arbitrary finite set of available bids $\B_i\subseteq\mathbb{R}_{\ge0}$, with $\bar B_i\defeq\max \B_i$ and $\bar B_i\ge\sup\V_i$.
    Moreover, we denote by $\B\defeq\times_{i\in[n]}\B_i$ the set of all possible joint bid profiles.
    A bidding strategy $\sigma_i$ for bidder $i$ is a (possibly randomized) mapping from their types $\V_i$ to their available bids $\B_i$. We represent such strategies as a $|\V_i|\times|\B_i|$ right stochastic matrix in which each row specifies a well-defined probability distribution over bids: bidder $i$'s strategy space is $\Sigma_i\defeq\mleft\{\sigma_i\in\mathbb{R}_{\ge0}^{|\V_i|\times|\B_i|}:\sigma_i\bb{1}=\bb{1}\mright\}$.
    We observe that bidders cannot condition their bids on their quality factors $\gammav$, since they are only known to the platform, and not to advertisers.
    Finally, we define the set of joint bidding strategies as 
    \[
    \Sigma\defeq \mleft\{\sigma\in\Delta_{\V\times\B}\,:\,\sum_{\bv\in\B}\sigma[\vvec,\bv]=\cF(\vvec),\,\,\forall\vvec\in\V\mright\}.\footnote{Given a matrix $M$, we denote by $M[i]$ its $i$-th row vector, and by $M[i,j]$ the entry in position $(i,j)$.}
    \]
    At each iteration $t$, bidder $i$ places bids according to a bidding strategy $\sigma_i^t\in\Sigma_i$. In particular, bidder $i$ observes its own type $v_i^t$, and then submits a bid $b_i^t\sim\sigma_i^t[v_i^t]$. 
    Then, bidder $i$ experiences a reward which we define as a function $u_i^t:\B_i\to\mathbb{R}$. The utility function $u_i^t$ observed at time $t$ implicitly depends on the realized vector of bids $\bv^t_{-i}$, and quality factors $\gammav^t$, and it is such that, for each possible bid $b\in\B_i$, 
    $
         u_i^t(b)\defeq 
            \alpha^t_{j,g(i)} \gamma^t_{g(i)} v^t_i - p_i((b,\bv_{-i}^t),\gammav),
    $
    with $j=\nu((b,\bv^t_{-i}),\bb{\gamma}^t,i)$.
    
    \paragraph{Regret and equilibria}
    Given a sequence of decisions $(b_i^1,\ldots,b_i^T)$ up to time $T$, the \emph{external regret} of bidder $i$ in type $v_i$ is how much they regret not having played the best fixed action in hindsight at each iteration in which they observed type $v_i$. Formally, the regret experienced by bidder $i$ under a certain type $v_i\in\V_i$ is
    \[
    R_{v_i}^T \defeq \max_{\hat b\in\B_i}\mleft\{\sum_{t=1}^T\one[v_i=v_i^t]\mleft(u_i^t(\hat b)-u_i^t(b_i^t)\mright)\mright\}.
    \]
    Then, the cumulative external regret of bidder $i$ at time $T$ is
  $
    R_i^T\defeq \sum_{v_i\in \V_i} R_{v_i}^T.
  $
    Let $(\bv^t)_{t=1}^T$ be the sequence of decisions made by the bidders up to time $T$.
    Then, the \emph{empirical frequency of play} $\bar\sigma^T\in\Delta_{\V\times\B}$ obtained from the realized sequence of types $(\vvec^t)_{t=1}^T$, and from the sequence of play $(\bv^t)_{t=1}^T$ is such that, for every $(\vvec,\bv)\in\V\times\B$:
    \[
    \bar\sigma^T[\vvec,\bv]\defeq \frac{1}{T}|\mleft\{1\le t\le T: \bv^t=\bv, \vvec^t=\vvec\mright\}|.
    \]
    If each bidder $i$ plays so as to obtain a regret $R_i^T$ growing sublinearly in $T$, then, in the limit as $T\to\infty$, the empirical frequency of play $\bar\sigma^T$ is guaranteed to converge almost surely to a \emph{Bayesian coarse correlated equilibrium} (BCCE)~\cite{CKK15,hartline2015no}.%
    We denote by $\bcce\subseteq\Sigma$ the set of all Bayesian coarse correlated equilibria of the game.

\section{Group Fairness in GSP Auctions}\label{sec:fair gsp auctions}

In this section we present the two group fair division schemes that will be added as a post-auction layer to GSP.

\paragraph{Preliminary definitions} Consider an arbitrary stage $t$ of the repeated auctions process (dependence on $t$ will be omitted when clear from the context).
For each group $w\in\gr$, let $\alg_w(\bv)\defeq \sum_{i\in I_w}\gamma_w\alpha_{\nu(i),w}b_i$ be the value obtained by group $w$ via a generic mechanism with allocation rule $\nu$, on bid vector $\bv$. Since $\gr=\{h,\ell\}$, the overall value is  $\alg(\bv)= \alg_h(\bv) + \alg_\ell(\bv)$.
Moreover, given a set of slots $J'\subseteq [m]$ assigned to group $w\in\gr$, we define $$\alg_w(J',\bv) \defeq \sum_{j\in [|J'|]}\gamma_{w} \alpha_{J'[j],w} b_{I_w[j]},$$
where $J'[j]$ denotes the slot with the $j$-th click-through rate among slots in $J'$, and $I_w[j]$ is the bidder belonging to $I_w$ with the $j$-th effective bid in decreasing order (e.g., $I_h(1)$ is the bidder of group $h$ with the highest effective bid).
Intuitively, $\alg_w(J',\bv)$ is the maximum value attainable by group $w$ when it is allocated $J'$ and the bid vector is $\bv$. 

\paragraph{$\beta$-EF1 mechanism}
The first mechanism that we describe employs a fair division scheme that guarantees the following notion of group fairness. %
\begin{definition}(Group $\beta$-EF1 fairness)\label{def:ef1}
    Let $\beta\defeq \xi_\ell/\xi_h$, with $\xi_h\in\mathbb{N}^+$ and $\xi_\ell =1$,  %
    \add{$\xi_h+\xi_\ell\le m$}.
	We say that an allocation is {\em group $\beta$ envy-free up to one good} ($\beta$-\emph{EF1 fair}) for $\beta\le 1$ and bid profile $\bv$ if, for each pair of groups $h,\ell \in \gr$, there exists one item $j_h\in J_h$ such that
	$
	\alg_\ell(J_\ell,\bv)\ge \beta\, \alg_\ell
(J_h \setminus \{j_h\},\bv).
	$
\end{definition}

\noindent
A group $\beta$-EF1 fair allocation can be obtained through a round robin procedure that assigns the first slot to group $h$, the second slot to group $\ell$, and then, in each of the following rounds, $\xi_h$ slots to group $h$ and $\xi_\ell=1$ slots to group $\ell$. The proof of this result is very similar to the one  for the classical EF1 fair division scheme by Markakis~\shortcite{Markakis17}. 
As an example, if we assume group $h$ to be the majority group (i.e., advertisers from group $h$ are allocated the highest slot in the ranking), then the result of the application of the group $\beta$-EF1 round robin procedure is the shift of advertisers assigned to a position $j\in J_h$ to position at most $\lceil (1+\beta) j\rceil -1$.
For more details on how this fair division scheme is implemented see the extended version of the paper.
We observe that a round robin procedure guarantees group EF1 even when the number of groups is $|\gr|>2$. %
The reason is that, for any pair $w,w'\in\gr$, and positions $j\in[|J_w|],j'\in[|J_{w'}|]$, with $j\le j'$, it holds $\alpha_{J_w[j],w}\ge\alpha_{J_{w'}[j'],w}$. Moreover, it is possible to show that a round robin procedure also guarantees group $\beta$-EF1 with respect to valuations (proofs can be found in the extended version of the paper).
\begin{restatable}{theorem}{efoneval}\label{thm:ef1 val}
Given a bid profile $\bv$, the allocation computed by the composite mechanism is $\beta$-EF1 fair with respect to the valuation profile $\vv$. In particular, given the allocation of slots to the two groups $J_h,J_\ell$, for each pair of groups $h,\ell\in \gr$, there exists one item $j\in J_\ell$ such that $\alg_h(J_h,\vv)\ge \beta\, \alg_h(J_\ell \setminus \{j\},\vv)$.
\end{restatable}

\noindent Intuitively, this is because the $\beta$-EF1 mechanism computes $J_h,J_\ell$ without employing the reported bid profile, which is used only to determine the per-group ranking.

\paragraph{$\beta$-EFX mechanism} The second notion of fairness which we consider is group $\beta$-EFX fairness. 

\begin{definition}(Group $\beta$-EFX fairness)\label{def:efx}
	An allocation is {\em group $\beta$-envy free up to any good} ($\beta$-{\em EFX fair}) for $\beta\le 1$ and bid profile $\bv$ if, for each pair of groups $h,\ell\in \gr$, and for each item $j_h\in J_h$, it holds
	$
	\alg_\ell(J_\ell,\bv)\ge \beta\, \alg_\ell(J_h \setminus \{j_h\},\bv).
	$
\end{definition}

\noindent
A group $\beta$-EFX fair division scheme can be obtained through a ``group version'' of the envy-cycle elimination algorithm by Lipton \emph{et al.}~\shortcite{lipton2004approximately}.
In particular, we propose the \emph{Group Envy-Cycle-Elimination algorithm} (GECE). The GECE algorithm can be summarized as follows: denote by $J_h$ and $J_\ell$ the set of slots assigned, respectively, to group $h$ and $\ell$. We say that group $h$ envies group $\ell$ if $\alg_h(J_h,\bv)<\beta \alg_h(J_\ell,\bv)$. 
Initially, all the slots are not assigned, that is, $J_h=J_\ell=\varnothing$.
Then, the algorithm iterates through the slots in decreasing order of click-through rate. The first slot is assigned to group $h$. For each subsequent slot $j$, the algorithm checks if groups $h$ and $\ell$ envy each other, and, if this is the case, the algorithm swaps their allocations.
Otherwise, if group $\ell$ does not envy group $h$, then the next slot is assigned to group $h$, else, if $\ell$ envies $h$, the slot is assigned to group $\ell$.  
In the following theorem, we prove that the GECE algorithm is guaranteed to obtain a $\beta$-EFX allocation. 

\begin{restatable}{theorem}{EFXgece}\label{th:gece}
The allocation computed by the group envy-cycle-elimination (GECE) algorithm is group $\beta$-EFX fair. 
\end{restatable}

\noindent It is possible to show that a variation of GECE %
yields a group EFX fair allocation (i.e., $\beta$-EFX with $\beta=1$) even with more than 2 groups.

\begin{restatable}{corollary}{kgece}\label{th:k-gece}
The allocation computed by the k-group envy-cycle-elimination (k-GECE) algorithm is group EFX.  
\end{restatable}

\section{Efficiency and Budget Balance}\label{sec:efficiency and bb}

In this section, we study the efficiency and budget balance of the two mechanisms obtained by combining GSP with the two fair division schemes described in~\cref{sec:fair gsp auctions}. Let us denote one such composite mechanism by $\textsc{c}$.
The post-auction layer of the composite mechanism $\textsc{c}$ is modifying the GSP allocation of slots to bidders. Ideally, no bidder should be penalized for this re-allocation. Therefore, we need to update the payments so that bidders' utility is not negatively affected by the composition of GSP with the fair division scheme.
Interestingly, we can do so starting from the payments of GSP. In particular, denote by $p_i^\textsc{g}$ and by $p_i^\textsc{c}$, the payments charged to advertiser $i$ computed by GSP and by the composite mechanism, respectively. Moreover, let $\nu^\textsc{g}(i)$ and $\nu^\textsc{c}(i)$ be the slots assigned to advertiser $i$ by GSP and by the composite mechanism, respectively.
Then, we define the payments charged by the composite mechanism $\textsc{c}$ as:
\begin{equation}\label{eq:payments}
         p_i^\textsc{c}\hspace{-1mm}\defeq \hspace{-1mm}\mleft\{\hspace{-1.25mm}\begin{array}{l}
            \displaystyle
            p_i^\textsc{g} \hspace{3.6cm}\textnormal{ if }\hspace{.2cm} \nu^\textsc{c}(i)\leq \nu^\textsc{g}(i) 
            \\[5mm]
            \displaystyle 
            p_i^\textsc{g}- 2 b_i \, \gamma_{g(i)}
            \mleft(\alpha_{\nu^\textsc{g}(i),g(i)} - \alpha_{\nu^\textsc{c}(i),g(i)} \mright) \hspace{0.3cm}\textnormal{ else.}
        \end{array}\mright.
\end{equation}
The pricing rule shows that the composite mechanism \textsc{c} compensates the loss of social welfare of the advertisers that obtain a worse slot by reducing their payments. In order to ensure individual rationality, the advertisers that obtain a better slot are not asked to compensate with a higher payment. Observe that the pricing rule does not exclude positive transfers to the advertisers.

Now, let us first define an appropriate notion of budget balance for a composite mechanism which assigns a payment $p_i^\textsc{c}$ to bidder $i\in \I$, with respect to the GSP mechanism. Let $p^\textsc{g}\defeq \sum_{i\in \I} p_i^\textsc{g}$, and $p^\textsc{c}\defeq \sum_{i\in \I} p_i^\textsc{c}$.

\begin{definition}
A composite mechanism is $\alpha$-budget balanced, for $\alpha\geq 0$, if 
$
p^\textsc{g} - p^\textsc{c} \leq \alpha \, p^\textsc{g}.
$
\end{definition}
\noindent An $\alpha$-budget balanced mechanism is therefore able to cover via the GSP payments at least an $\alpha$ fraction of the compensations given to the bidders by the composite mechanisms.  

Let $\alg^\textsc{g}(\bv)$ and $\alg^\textsc{c}(\bv)$ be, respectively,  the value of the GSP mechanism, and of the fair composite mechanism on an arbitrary bid vector $\bv$.  
Moreover, we use the following two assumptions in the analysis of the composite mechanisms.
\begin{assumption}
\label{assumption1}
The value of the minority group  increases after the application of the composite mechanism, i.e., $\alg^\textsc{c}_\ell(\bv) \geq \alg^\textsc{g}_\ell(\bv).$
\end{assumption}

\begin{assumption}
\label{assumption2}
The first slot is assigned by GSP to $I_h(1)$, i.e., the first  bidder of group $h$.
\end{assumption}

\noindent \cref{assumption1} is natural since the basic goal of the proposed mechanisms would be to make the minority group better off with respect to the GSP case. 
\add{\cref{assumption2} only requires that the advertiser with the highest bid belongs to the majority group. This is natural, for example, in settings where groups have different economic power.}
First, we provide efficiency and budget balance results for the $\beta$-Fair GSP mechanism. 

\begin{restatable}{theorem}{efficiencyOne}
\label{thm:swl_ref}
The $\beta$-Fair GSP mechanism achieves a value that is at least a $1/(1+\beta)$ fraction of the value of GSP, i.e., for all bid vectors $\bv\in\B$, $\alg^\textsc{c}(\bv) \geq \alg^\textsc{g}(\bv)/(1+\beta).$
\end{restatable}

\begin{restatable}{theorem}{bbOne}\label{th:bb beta fair gsp}
The $\beta$-Fair GSP mechanism is $2$-budget balance. 
\end{restatable}
\noindent\add{
This result does not rule out that the mechanism may suffer a net loss. However, such loss is bounded by a small constant (see \cref{sec:exp} for an empirical evaluation of such loss).
}

Second, we study efficiency and budget balance of the GSP-EFX mechanism. The allocation done by the mechanism is described in Section \ref{sec:fair gsp auctions} and it is obtained by the composition of the GSP mechanism with the group EFX fair division scheme.  The payments of GSP-EFX are computed as in~\cref{eq:payments}. %
Then, we have the following.
\begin{restatable}{theorem}{ALGEfx}
\label{thm:rel_GSP_EFX}
 GSP-EFX achieves a value that is at least a fraction $1/3$ of the value of  GSP, i.e., for all bid vectors $\bv\in\B$, $\alg^\textsc{c}(\bv) \geq \frac{1}{3} \alg^\textsc{g}(\bv)$.
\end{restatable}
Finally, we prove that GSP-EFX is able to compensate at least $1/4$ of the total welfare loss generated by the application of the EFX fair division scheme. Formally,
\begin{restatable}{theorem}{bbEfx}
The GSP-EFX mechanism is 4-budget balance.
\label{thm:payment_compensation_EFX}
\end{restatable}
\begin{remark}
Given a number of groups $k>2$ and $\beta=1$, \cref{thm:swl_ref} and \cref{thm:rel_GSP_EFX} can be extended to show that the group EF1 (resp., group EFX) mechanism achieves a value that is at least a 1/2k fraction (resp., a fraction 1/3k) of the value of GSP.
\end{remark}

\section{Price of Composition}\label{sec:poc}

Equipped with the results from~\cref{sec:efficiency and bb} we can study the performance of the proposed mechanisms at equilibrium. 
We study the quality of the equilibria emerging as the results of the no-regret learning dynamics in which each bidder behaves as an external-regret minimizer. To do so, we propose the \emph{price of composition} (PoC) as a natural measure to evaluate the social welfare guarantee of our mechanisms at equilibrium.
For an arbitrary mechanism, the social welfare attained for bids $\bv$, valuations $\vvec\in\V$, and quality factors $\bb{\gamma}$ is
\begin{equation}\label{eq:social welfare}
    SW(\bv,\vvec,\bb{\gamma})\defeq \sum_{j\in [m]}\alpha_{j,g(\pi(j))}\gamma_{g(\pi(j))}v_{\pi(j)}.
\end{equation}

\noindent
Given an equilibrium strategy $\sigma\in\bcce$ in an incomplete-information game, its social welfare is evaluated by comparing it to the expected ex-post social welfare of the GSP mechanism, which we denote by $\expe_{\vvec,\gammav}[SW^\textsc{g}(\vvec,\gammav)]$. In particular, we can define the following worst-case ratio.
\begin{definition}\label{def:poc}
The \emph{price of composition (PoC)} of a composite mechanism \textsc{c} is defined as 
\begin{equation*}
\poa\defeq\inf_{\cF,\cG,\sigma\in\bcce} \frac{\expe_{\vvec,\gammav,\bv\sim\sigma}[SW^\textsc{C}(\bv,\vvec,\gammav)]}{\expe_{\vvec,\bb{\gamma}}[\textsc{SW}^\textsc{G}(\vvec,\bb{\gamma})]}. 
\end{equation*}
\end{definition}
\noindent
By the definition of the composite mechanisms, their social-welfare is at most equal to the social welfare of the GSP mechanism (see~\cref{sec:efficiency and bb}), i.e., $\poa \in [0, 1]$.
Moreover, bounding the worst case $\poa$ automatically yields a $\poa/2$ guarantee on the \emph{price of total anarchy} \cite{blum2007learning}. This is because $SW^\textsc{g}(\vvec,\gammav)\ge  SW^\ast(\vvec,\gammav)/2$, where $SW^\ast$ is the optimal social welfare with valuations $\vvec$.~\footnote{The value of the GSP solution for the Ad-Types problem with group-specific CTRs is at least $1/2$ of the optimal solution provided by maximum weighted matching \cite{Colini-Baldeschi20}.}

In order to characterize the PoC of our mechanisms, we introduce a natural smoothness condition for composite mechanisms, which is a generalization of the notion of semi-smoothness by Lucier and Paes Leme~\shortcite{lucier2011gsp}.
Let $SW^\textsc{G}$ be the social welfare of the baseline mechanism (that is, in our setting,  GSP), and let $SW^\textsc{c}$ be the social welfare provided by the composite mechanism (in our case, $\beta$-Fair GSP or GSP-EFX) given an arbitrary vector or valuations. 
Then, our notion of smoothness states that social welfare gaps between the social welfare at the baseline mechanism with truthful bid vector $\vvec$, and the social welfare for an arbitrary joint strategy profile $\sigma\in\Sigma$ in the composite mechanism, can be captured by the marginal increases in the individual agents' utilities when unilaterally switching to a deviation strategy profile $\sigma'=(\sigma_1',\ldots,\sigma_n')$.
Formally, a composite mechanism \textsc{c} is \emph{$(\lambda,\mu)$-semi-smooth with respect to a baseline mechanism \textsc{g}} if there exists a profile of individual bidding strategies $\sigma'=(\sigma_1',\ldots,\sigma_n')\in\bigtimes\Sigma_i$ such that, for any joint strategy $\sigma\in\Sigma$, $\vvec\in\V$, and $\gammav$, 
\begin{multline}\label{def:smoothness2}
    \displaystyle
    \expe_{\substack{\bv\sim\sigma[\vvec]\\\bv'\sim\sigma'[\vv]}} \big[\sum_{i\in [n]} u^\textsc{c}_i((b_i',\bv_{-i}), v_i,\gammav)\big]\\\ge \lambda SW^\textsc{g}(\vvec,\gammav)-\mu \expe_{\bv\sim\sigma[\vvec]}\mleft[SW^\textsc{c}(\bv,\vvec,\gammav)\mright].
\end{multline}
In this setting, we make the additional assumption that bidders are conservative and, therefore, they do not overbid. This is in line with previous works studying the price of anarchy of auctions with conservative bidders~(see, e.g., \cite{PT10,CKS16,Feldman2013SimultaneousAA}).~\footnote{
Our bounds on PoC can be adapted with minor modifications to the case in which bidders are $\delta$-conservative~\cite{BR11}, i.e., $b_i\leq \delta v_i, \forall i$. In particular, all price of composition factors are multiplied by an additional $\delta$ factor.}  
Under this assumption, the following holds.

\begin{restatable}{lemma}{smooth}\label{lemma:semismooth}
	The $\beta$-Fair GSP mechanism is $\mleft(1/2,(1+\beta) \mright)$-semi-smooth, and the GSP-EFX mechanism is  $\mleft(1/2, 3 \mright)$-semi-smooth. Both mechanisms are individually rational at the equilibrium. 
\end{restatable}

\noindent By the fact that $(\lambda, \mu)$-semi smoothness implies a $\frac{\lambda}{1+\mu}$ $\poa$, we characterize the $\poa$ of the mechanisms as follows.
\begin{restatable}{theorem}{pocOne}\label{thm:poa beta fgsp}
The price of composition with uncertainty of $\beta$-Fair GSP is 
$
\poa= (2(2+\beta))^{-1}.
$
\end{restatable}

\noindent This result shows that, even if the bidders are self-interested (i.e., they take decisions so as to minimize their own individual regret), our composite mechanism can guarantee \emph{group fairness}, as well as convergence to a \emph{good} equilibrium point. %
Analogously, it is possible to prove the following for GSP-EFX:
\begin{restatable}{theorem}{pocTwo}\label{thm:poa_EFX_GSP} The price of composition with uncertainty of GSP-EFX is 
$\poa= 1/8$.
\end{restatable}

\begin{figure}[ht]
\centering
\hspace{-.4cm}
\begin{minipage}{.23\textwidth}
\includegraphics[scale=0.3]{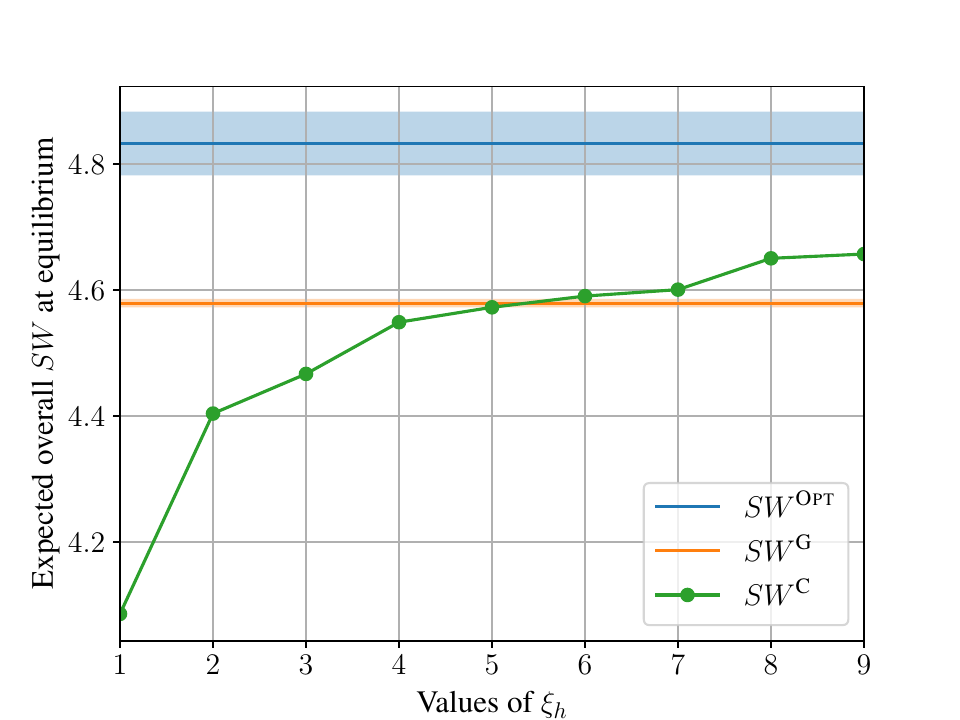}
\end{minipage}
\hspace{.2cm}
\begin{minipage}{.23\textwidth}
\includegraphics[scale=0.3]{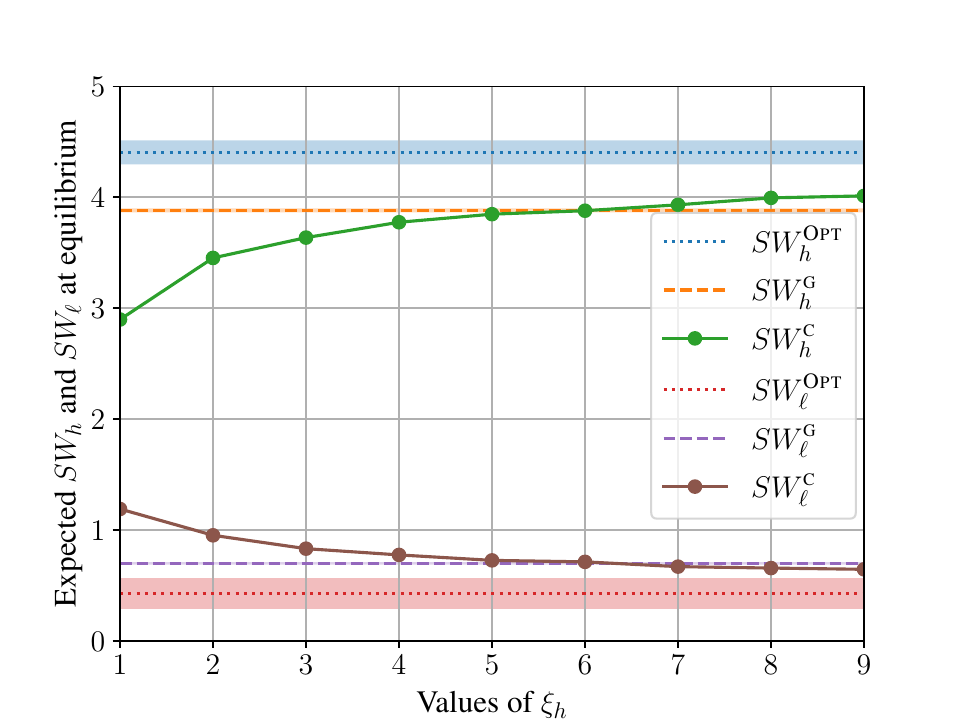}
\end{minipage}

\begin{minipage}{.3\textwidth}
\includegraphics[scale=0.32]{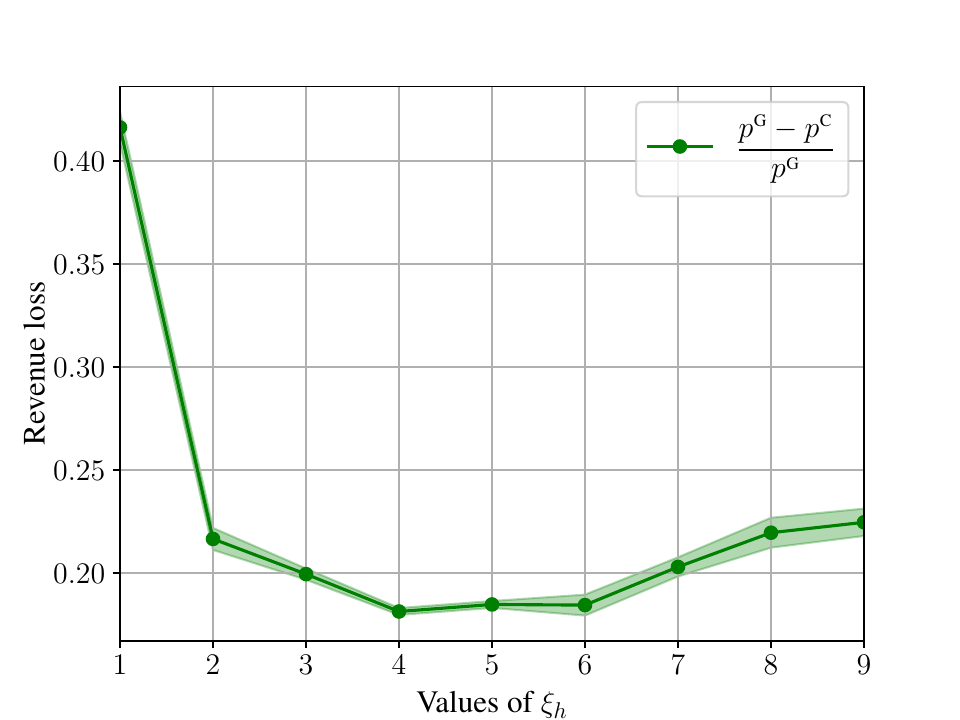}
\end{minipage}
\caption{\textbf{(Top-Left):} Expected overall SW at equilibrium. \textbf{(Top-Right):} expected per-group SW at equilibrium. \textbf{(Bottom):} expected fraction of the GSP payments to be used for compensation in the composite mechanism. We set $\xi_\ell=1$ and we let $\xi_h$ vary.}
\label{fig:exp}
\end{figure}

\section{Experimental Evaluation}\label{sec:exp}

We experimentally evaluate the quality of the equilibria emerging as the results of no-regret learning dynamics in which agents interact through the $\beta$-Fair GSP mechanism.

\paragraph{Regret minimization for the Bayesian setting}
For each bidder $i$, we consider a discrete set of bids $\B_i$. 
We focus on the partial-information setting, in which, at each time instant $t$, each bidder observes only the reward $u_i^t(b_i^t)$ associated to its choice $b_i^t\in\B_i$. This is in line with what happens in real-world sponsored search auctions, where advertisers do not observe competing bids (i.e., they cannot compute a counterfactual utility $u_i^t(b)$ for each $b\in\B_i$), but can only observe the outcome associated to their decision $b_i^t$. 
We instantiate an external-regret minimizer $\mathcal{R}_{i,v}$ for each bidder $i$ and $v\in\V_i$. In practice, we use the \textsc{Exp3} algorithm by~\cite{auer2002nonstochastic} for each external regret minimizer. Then, we build a regret minimizer for the Bayesian bidder $i$ as follows: at each $t$, bidder $i$ observes their realized type $v^t\in\V_i$ and selects $b_i^t\in\B_i$ according to $\mathcal{R}_{i,v_i}$. Then, after utility $u_i^t(b_i^t)$ is observed, only $\mathcal{R}_{i,v_i}$ is updated. This simple procedure guarantees that $\limsup_{T\to\infty} R_i^T/T\le 0$, which implies that the empirical frequency of play $\bar\sigma^T$ of the dynamic converges almost surely in the limit to a Bayesian coarse correlated equilibrium (see, e.g.,~\cite[Lemma 10]{hartline2015no}).

\paragraph{Experimental setting}
We construct a real-world dataset through logs of a large Internet advertising company. 
We test our $\beta$-Fair GSP mechanism in an artificial environment where we have 20 advertisers, equally distributed among two groups, and competing for ad opportunities over a sequence of  $T=10^4$ auctions.
For each $i$, we set $\V_i=\{x/100: x\in[100]\cup\{0\}\}$. Moreover, in order to be able to create a steep unbalance between the two groups, we let, for each $\vvec\in\V$, $\cF(\vvec)=\cF_1(v_1)\cdot\ldots\cdot\cF_n(v_n)$, with $\cF_i\in\Delta_{\V_i}$ for each $i$. Then, for each $i\in I_h$ (i.e., advertiser $i$ belongs to the majority group), we artificially set value distributions to be such that $\cF_i(1)=1$, and $\cF_i(v)=0$ for each $v\in\V_i\setminus\{1\}$. Each bidder $i\in I_\ell$ has a value distribution $\cF_i$ built by normalizing the distribution of bids observed from real-world bidding data of a large Internet advertising company. \add{We use bids as a proxy for true valuations of advertisers.}
Discount curves are computed by averaging and normalizing in $[0,1]$ real-world discount factors. In particular, we estimate discount curves on ads optimizing for two distinct conversion types, one per group. This models different per-group preferences on the slots. Quality factors are set to be $\gammav=(1,1)$. Experiments are run on a 24-core machine with 57Gb of RAM.

\paragraph{Results}
Each advertiser takes decisions so as to minimize their external regret according to the procedure described above.
Then, we analyse the equilibrium outcomes originating from advertisers behaving so as to minimize their regret, and interacting through GSP and $\beta$-Fair GSP, respectively. To do so, we compute the empirical frequency of play $\bar\sigma^T$ in the two settings. Let $\bar\sigma^{T,\textsc{g}}$ and $\bar\sigma^{T,\textsc{c}}$ be the empirical frequency of play obtained via the GSP mechanism, and the empirical frequency of play obtained via the $\beta$-Fair GSP mechanism, respectively.
\cref{fig:exp}--(Top-Left), and \cref{fig:exp}--(Top-Right) report a comparison between the expected social welfare attained at $\bar\sigma^{T,\textsc{g}}$ (i.e., $SW^\textsc{g}$), the expected social welfare attained at $\bar\sigma^{T,\textsc{c}}$ (i.e., $SW^\textsc{c}$) for different values of $\xi_h$, while for simplicity we keep $\xi_\ell=1$ (see~\cref{def:ef1}), and the expected social welfare obtained when advertisers are not strategic and submit bids truthfully to a GSP mechanism (i.e., $SW^{\textsc{Opt}}$). Each value is computed over $20$ repetitions of the dynamics, and figures display the resulting mean and standard deviation. 
In particular, \cref{fig:exp}--(Top-Left) reports the overall social welfare, and shows that, as expected, lower values of $k$ increase $\Delta SW$, that is, the gap between the social welfare provided by the GSP solution and the social welfare provided by the $\beta$-Fair GSP solution. The gap is negligible for most values of $\beta=1/\xi_h$, and in the worst case, the empirical gap is approximately $10\%$ of the social welfare provided by GSP. 
\cref{fig:exp}--(Top-Right) reports the per-group social welfare at equilibrium. 
Finally, \cref{fig:exp}--(Bottom) describes the fraction of the revenue which is lost due to compensation for advertisers as defined in~\cref{eq:payments}.
We observe that, in practice, the fraction of the GSP prices used for compensation significantly lower than the worst case bound of~\cref{th:bb beta fair gsp}. In particular, the composite mechanism looses nearly $40\%$ of GSP revenue when $\xi_h=1$ (i.e., when the group fairness constraint is as tight as possible). However, for higher values of $\xi_h$, the fraction of $p^\textsc{g}$ which has to be used for compensations stabilizes around $20\%$. This suggests that there may be a trade-off between guarantees for advertisers and revenue losses incurred by the platform where this type of mechanisms could be viable in practice.

\section{Discussion and Future Works}

We proposed a class of composite mechanisms constructed through the combination of GSP with different fair division schemes. We show that these composite mechanisms yield group fairness guarantees for advertisers, and we characterized the costs which the platform has to incur. 
In the future, it would be interesting to study fairness guarantees in the space of valuations, and in the space of utilities. We already provided some preliminary results going in this direction (\cref{thm:ef1 val}). Extending this to EFX presents non-trivial additional challenges, such as showing that there exists a monotone version of GECE. We leave this as a future research direction.

\section*{Acknowledgements}
The Sapienza team was supported by the ERC Advanced 
Grant 788893 AMDROMA ``Algorithmic and Mechanism Design Research in 
Online Markets'',   the MIUR PRIN project ALGADIMAR ``Algorithms, Games, and Digital Markets'', and the Meta Research grant on ``Fairness and Mechanism Design''. 

\bibliographystyle{named}
\bibliography{refs}

\clearpage
\appendix
\onecolumn
\begin{center}
    \fbox{\LARGE{\bfseries ~Supplementary Material}}
\end{center}
\vspace{1cm}

\section{Related Works}\label{sec: related works}

The process of skewing ads delivery on the basis of particular demographic characteristics has recently been studied in several works \cite{ali2019discrimination,ali2019ad,lambrecht2018algorithmic}.  \cite{gelauff2020advertising} study the problem of advertising for demographically representative outcomes, with the goal of seeking fairness in the realized conversions generated by the advertising campaigns. \cite{nasr2020bidding} present bidding strategies that advertisers can use to avoid gender discrimination at the user level.  
 Moreover, a recent line of work studies how to design advertising mechanisms that ensure diversity and fairness between users. In this context, individual fairness \cite{dwork2012fairness} and group fairness \cite{FFMSV15} are the most widely adopted paradigms.~\cite{ilvento2020multi} propose a model of fairness that combines the complementary notions of envy-freeness and individual fairness, with the goal of maximizing the platform utility defined as the sum of the advertiser's bids. Moreover, \cite{chawla2020fairness} study the trade-offs between social welfare maximization and fairness in the context of ad auctions, while~\cite{DworkI19} address the problem of designing truthful ad auctions that satisfy individual fairness. The work by~\cite{DworkIJ20} investigates individual fairness under pipeline composition, motivated by the fact that a system built from individually fair components may not itself be individually fair. Finally, \cite{celis2019toward} propose an optimal Bayesian auction to maximize the platform's revenue conditioned on ensuring that the audience seeing an advertiser's ad is distributed appropriately across sensitive types such as gender or race.

Fairness in online learning for ad auctions has also been studied within the the contextual stochastic multi-armed bandit framework.  Specifically, there is a set of arms/advertisers described by latent vectors which is drawn from an unknown distribution. At each round a new user is drawn from the distributions of the contexts of two groups, a majority group and a minority group. \cite{NIPS2016_6355} study the problem of designing a learning algorithm that ensures individual fairness between the arms/advertisers. Moreover, \cite{pmlr-v75-raghavan18a} analyse the problem of designing a bandit algorithm that achieves a regret for each group that is as close as possible to the regret experienced when the group plays in isolation under the same algorithm.

Envy-freeness has been extensively studied in fair division \cite{foley1967,Varian} and in online advertising \cite{colini2014revenue,colini2016revenue,DBLP:journals/corr/BranzeiFMZ16,chen2010envy,guruswami2005profit}.
Unfortunately, envy-freeness does not always admit a solution for indivisible items and, therefore, various relaxations of this concept have been proposed in literature.
The notion of envy-freeness up to one good (EF1) has been introduced by~\cite{Budish11},
and in the work by~\cite{lipton2004approximately}.~\cite{Budish11} also defines the notion of maximin share allocation (MMS), based on concepts introduced by \cite{Moulin90}. Then,~\cite{CaragiannisKMPS19} introduced the notions of envy freeness up to any good (EFX) and pairwise maximin share allocation (PMMS). More recently,~\cite{BMV18}  proposed to study groupwise maximin share fairness (GMMS) allocations.  
The notion of group EF1 allocation has recently been studied in \cite{kyropoulou2020almost}, whereas the notion of approximation in several fair division schemes has been explored in \cite{AMN20}. Fair division problems have been studied in many different settings: 
from apartment renting~\cite{gal2017fairest} to allocating blood donations to blood banks~\cite{mcelfresh2020matching}, and in other applications ~\cite{brams1996procedure,ghodsi2011dominant,kash2014no,kroer2019scalable,aleksandrov2015online}.

The price of anarchy \cite{KP09} of GSP equilibria have been studied in various works (see, e.g.,~\cite{lucier2011gsp,roughgarden2015intrinsic,CKK15}). The price of anarchy of an auction is traditionally measured through the overall social welfare attained at the equilibrium (see, e.g., reference \cite{CKK15} and related work), as it is the price of composition that we introduce in this work. However, prices are not completely ignored in our work, as it is the case with standard GSP auctions. In our setting monetary incentives are crucial to  select equilibria with good price of composition.

Finally, given that no-external regret learning dynamics converge to the set of coarse correlated equilibria (see~\cite{blum2007learning,cesa2006prediction}),~\cite{blum2008regret} study the price of \emph{total} anarchy in games in which players take decisions so as to minimize their external regret. 
No-regret learning dynamics converging to correlated equilibria have been studied in various settings (see, e.g., \cite{foster1997calibrated,hart2000simple,celli2019learning,celli2020noregret,farina2021simple}).
Moreover, \cite{hartline2015no} and~\cite{CKK15} study the quality of outcomes emerging from no-regret dynamics in Bayesian settings.
    
\iftrue
\section{Summary of the Notation}

\cref{tab:notation} summarizes the main notation used throughout the paper. 

\begin{table}[htp]
\centering
    \begin{tabular}{p{1.5cm}p{15cm}}
    \toprule
        \centering\textbf{Symbol} & \textbf{Description}\\
    \midrule
        \centering $\V_i$ & Set of types of bidder $i$.\\
        \centering $\V$ & Set of joint bidders' types, $\V\defeq\bigtimes\V_i$.\\
        \centering $\cF$ & Probability distribution over $\V$.\\
        \centering $\cG$ & Distribution of quality factors.\\
    \midrule 
        \centering $\gr$ & Set of bidders' groups.\\
        \centering $g$ & Function $g:[n]\to\gr$ such that $g(i)$ is the group of bidder $i$.\\
        \centering $I_{(\cdot)}$ & Set of bidders (i.e., $I_{(\cdot)}\subseteq [n]$) belonging to group $(\cdot)\in\gr$.\\
        \centering $I_{(\cdot)}[j]$ & Bidder with the $j$-th highest valuation among bidders of group $(\cdot)$.\\
        \centering $J_{(\cdot)}[j]$ & Slot with the $j$-th highest click-through rate among slots assigned to group $(\cdot)$.\\
    \midrule
        \centering$\B_i$ & Set of available bids of bidder $i$.\\
        \centering$\B$ & Set of joint bid profiles, $\B\defeq \bigtimes_{i\in[n]}\B_i$.\\
        \centering$\sigma_i$ & Bidding strategy of bidder $i$, $\sigma_i:\V_i\to\Delta_{\B_i}$.\\
        \centering$\Sigma_i$ & Set of individual bidding strategies of bidder $i$.\\
        \centering$\Sigma$ & Set of joint bidding strategies.\\
    \midrule 
        \centering $p_i(\bv,\gammav)$ & Price charged to bidder $i$ under bid vector $\bv$ and quality factors $\gammav$.\\
        & Abbrev.: $p_i$ when $\bv$ and $\gammav$ are clear from the context.\\
        \centering $\pi(\bv,\gammav,j)$ & Bidder assigned by the mechanism to slot $j$.\\
        & Abbrev.: $\pi(j)$ when clear from context.\\
        \centering $\nu(\bv,\gammav,i)$ & Slot assigned by the mechanism to bidder $i$.\\
        & Abbrev.: $\nu(i)$ when clear from context.\\
    \midrule
        \centering $u_i(\bv,\vvec,\gammav)$ & Utility observed by bidder $i$ under bids vector $\bv$, realized valuations $\vvec$, and quality factors $\gammav$.\\
        \centering $\alg(\bv)$ & Value obtained by the mechanism under bid vector $\bv$.\\
        \centering $SW(\bv,\vvec,\gammav)$ & Overall social welfare provided by the mechanism under bid vector $\bv$, realized valuations $\vvec$, and quality factors $\gammav$.\\
        \centering $SW_{(\cdot)}(J')$ & Social welfare for bidders in group $(\cdot)\in\gr$ under the best allocation of slots in $J'\subseteq [m]$ to bidders in $I_{(\cdot)}$.\\
    \bottomrule
    \end{tabular}
    \caption{Summary of notation used in the paper.}
    \label{tab:notation}
\end{table}

\fi

\clearpage
\section{Group Fair Division Schemes}
\label{app:fairalgorithms}

In this section, we report the the pseudocode for the two fair division schemes described in~\cref{sec:fair gsp auctions}.

\paragraph{$\beta$-Fair GSP mechanism.} The first fair division scheme (see Algorithm~\ref{alg:RR_EF1}) takes as input the two groups of advertisers $I_h$ and $I_\ell$,  the $\beta=\xi_\ell/\xi_h$ parameter, the GSP allocation $\pi^\textsc{g}$, and returns the allocation of $\beta$-EF1 GSP. 
Let $I_{(\cdot)}[j]$ be the bidder of group $(\cdot)\in\gr$ with the $j$-th highest effective bid within that group. This means that, given the GSP allocation, 
$\nu^\textsc{g}(I_h[1])\le \nu^\textsc{g}(I_h[2]))\le \ldots\le\nu^\textsc{g}(I_h[|I_h|])$, and analogously for group $\ell$.
The Round-Robin scheme allocates the first $\xi_h$ slots to the majority group $h$ (if there are enough bidders from group $h$, see Line~\ref{step: rr allocate group h}). Then, it proceeds by allocating $\xi_\ell$ slots to group $\ell$ (Line~\ref{step: rr allocate group ell}).
Bidders belonging to the same group are allocated following the GSP ordering (i.e., in descending order of effective bid (see Lines~\ref{step: rr allocation 1},~\ref{step: rr allocation 2},~\ref{step: rr allocation 3},~\ref{step: rr allocation 4}). 
Finally, when at least one of the two groups has been fully assigned, the algorithm completes the allocation by allocating the remaining slots to the group that still has bidders which have not been given a slot (Lines~\ref{step: rr fill h},~\ref{step: rr fill ell}).

\paragraph{$k$-Group Fair EF1 mechanism.} The pseudocode of the Round Robin algorithm for EF1 allocation of $k$ groups is given in Algorithm \ref{alg:k-RR_EF1}.

\begin{minipage}{.5\textwidth}
\begin{algorithm}[H]\small
\caption{Round-Robin $\beta$-EF1 allocation}
\begin{algorithmic}[1]
\State \textbf{Input}: Groups $I_h$ and $I_{\ell}$, $\beta=\xi_\ell/\xi_h$, $\xi_h\in\mathbb{N}^+$, $\xi_l =1$, GSP allocation $\pi^\textsc{g}$
\State $\pi^\textsc{c}(1) \gets I_h[1], \pi^\textsc{c}(1) \gets I_l[1]$
\State $j=3,\,i_h=2,\,i_\ell=2$
\While {$i_h\le |I_h|$ and  $i_\ell\le |I_{\ell}|$}
\State $l=0$
\While{$j\le n$ and $l<\xi_h$ and $i_h\le |I_h|$}\label{step: rr allocate group h}
\State $\pi^\textsc{c}(j) \gets I_h[i_h]$\label{step: rr allocation 1}
\State $i_h=i_h+1$,\, $j=j+1$,\, $l=l+1$
\EndWhile 
\If{$j\le n$ and $i_{\ell}\le |I_{\ell}|$}\label{step: rr allocate group ell}
\State $\pi^\textsc{c}(j) \gets I_\ell[i_{\ell}]$\label{step: rr allocation 2}
\State $i_\ell=i_\ell+1$,\, $j=j+1$,\,  
\EndIf 
\EndWhile 
\While {$j\le n$ and $i_h \le|I_h|$}\label{step: rr fill h}
\State $\pi^\textsc{c}(j) \gets I_h[i_h]$\label{step: rr allocation 3}
\State $i_h=i_h+1$,\, $j=j+1$
\EndWhile
\While {$j\le n$ and $i_\ell \le |I_\ell|$}\label{step: rr fill ell}
\State $\pi^\textsc{c}(j) \gets I_\ell[i_\ell]$\label{step: rr allocation 4}
\State $i_\ell=i_\ell+1$,\, $j=j+1$
\EndWhile
\Return allocation $\pi^\textsc{c}$
\end{algorithmic}
\label{alg:RR_EF1}
\end{algorithm}
\end{minipage}
\hspace{.5cm}
\begin{minipage}{.45\textwidth}
\begin{algorithm}[H]\small
\caption{$k$-Group Round-Robin EF1 allocation}
\begin{algorithmic}[1]
\State \textbf{Input}: Groups $I_w$, $w\in [k]$, GSP allocation $\pi^\textsc{g}$
\State $j=1,\,i_w=1\,\forall w\in [k],\,w=1$
\While {$j\leq n$}
\If{$i_w \le |I_w|$} 
\State $\pi^\textsc{c}(j) \gets I_w[i_w], i_w=i_w+1$
\EndIf
\If{$w < k$} 
\State{$w=w+1$}
\Else 
\State{$w=1$}
\EndIf 
\EndWhile
\Return allocation $\pi^\textsc{c}$
\end{algorithmic}
\label{alg:k-RR_EF1}
\end{algorithm}
\end{minipage}

\paragraph{Group envy-cycle elimination} The second fair division scheme (Algorithm~\ref{alg:GECE}) takes as input the  two groups of advertisers $I_h$ and $I_\ell$,  the $\beta$ parameter, the GSP allocation $\pi^\textsc{g}$, and returns the allocation obtained via the GECE algorithm to be employed by GSP-EFX.
First, the algorithm computes an allocation of slots to groups (Lines from~\ref{step:gece start phase 1} to~\ref{step:gece end phase 1}). Then, the algorithm proceeds by allocating slots within each group by following the GSP ordering (Lines from~\ref{step:gece start phase 2} to~\ref{step:gece end phase 2}).
Algorithm~\ref{alg:GECE} starts by assigning the first slot to group $h$, and the second slot to group $\ell$, respectively. 
Then, the algorithm iterates through the slots in decreasing order of click-through rate. For each subsequent slot $j$, the algorithm checks if groups $h$ and $\ell$ envy each other, and, if this is the case, the algorithm swaps their allocations (Line~\ref{step: gece swap}).
Otherwise, if group $\ell$ does not envy group $h$, then the next slot is assigned to group $h$, else, if $\ell$ envies $h$, the slot is assigned to group $\ell$.  
Finally, the allocation $\pi^\textsc{c}$ is built by assigning the per-group allocations $J_h$ and $J_\ell$ in decreasing order of effective bid (i.e., by following the GSP ordering within each group).

\paragraph{$k$-Group envy-cycle elimination} The last fair division scheme which we consider is \emph{$k$-group envy-cycle elimination} ($k$-GECE), reported in \cref{alg:k-GECE}. 
This algorithm is an extension of GECE to the case of $k$ groups. In particular, given the current assignment of slots to groups (defined by sets $J_1,\ldots,J_k$), the algorithm computes the envy graph $G$ (Line~\ref{line:envy graph}). In such graph, every node represent a group, and there is a directed edge from group $g$ to group $g'$ if and only if $g$ envies $g'$. Then, $k$-GECE determines an alternative allocation such that $G$ is acyclic (Lines~\ref{line:acyclic start}
--\ref{line:acyclic end}, see \cite[Lemma 2.2]{lipton2004approximately}).
Since $G$ is acyclic there must be at least one group $w\in [k]$ which is not envied by other groups. The current element $j$ is assigned to such group $w$.
Finally, the allocation $\pi^\textsc{c}$ is built by assigning the per-group allocations $J_w$, $w\in [k]$, in decreasing order of effective bid (i.e., by following the GSP ordering within each group).

\begin{minipage}{.4\textwidth}

\begin{algorithm}[H]\small
\caption{Group Envy-Cycle-Elimination}
\begin{algorithmic}[1]
\State{\textbf{Input}: Groups $I_h$ and $I_\ell$, $\beta\le 1$, GSP allocation $\pi^\textsc{g}$, bid profile $\bv$}
\State{$J_h\gets\{1\},\, J_\ell\gets\{2\}$}
\For{$j=3$ \textbf{to} $n$}\label{step:gece start phase 1}
\If{$\alg_h(J_h,\bv)<\beta \alg_h(J_\ell,\bv)$ and $\alg_\ell(J_\ell,\bv)<\beta \alg_\ell(J_h,\bv)$}
\State{exchange $J_h$ with $J_\ell$}\label{step: gece swap}
\EndIf
\If{$\alg_\ell(J_\ell,\bv)\geq\beta \alg_\ell(J_h,\bv)$}
\State{$J_h = J_h \cup \{j\}$}
\Else
\State{$J_\ell = J_\ell \cup \{j\}$}\label{step:gece end phase 1}
\EndIf 
\EndFor
\State{$i_h=1,\,i_\ell=1$}\label{step:gece start phase 2}
\For{$j\in[n]$}
\If{$j\in J_h$}
\State $\pi^\textsc{c}(j)\gets I_h[i_h]$
\State $i_h=i_h+1$
\Else
\State $\pi^\textsc{c}(j)\gets I_\ell[i_\ell]$
\State $i_\ell=i_\ell+1$\label{step:gece end phase 2}
\EndIf
\EndFor
\Return{allocation $\pi^\textsc{c}$}
\end{algorithmic}
\label{alg:GECE}
\end{algorithm}
\end{minipage}
\hspace{.5cm}
\begin{minipage}{.5\textwidth}
\vspace{-1.2cm}
\begin{algorithm}[H]\small
\caption{$k$-Group Envy-Cycle Elimination}
\begin{algorithmic}[1]
\State{\textbf{Input}: Groups $I_1,\ldots,I_k$, item set $J$, bid profile $\bv$}
\State{$J_w\gets\{w\}, \forall w \in [k]$}\label{step: gece init}
\For{$j=k+1$ to $n$}
\State{Construct the envy graph $G$ for the current allocation $(J_w)_{w=1}^k$}\label{line:envy graph}
\While{There is no node of in-degree $0$ in $G$}\label{line:acyclic start}
\State{Find a cycle $i_1\leftarrow i_2, \ldots, \leftarrow i_r$ in $G$}
\State{$B=J_{i_1}$}
\For{$k=1$ to $r-1$}
\State{$J_{i_k}=J_{i_{k+1}}$}
\EndFor
\State{$J_{i_r}=B$}
\State{Update the envy graph $G$}\label{line:acyclic end}
\EndWhile
\State{Let $w\in[k]$ be a group of in-degree 0 in $G$}
\State{$J_w = J_w \cup \{j\}$}
\EndFor
\State{Compute the allocation $\pi^\textsc{c}$ analogously to Lines~\ref{step:gece start phase 2}--\ref{step:gece end phase 2} of \cref{alg:GECE}}
\State{\textbf{return} allocation $\pi^\textsc{c}$}
\end{algorithmic}
\label{alg:k-GECE}
\end{algorithm}
\end{minipage}

\section{Standard Definitions}\label{sec:appendix details}

\subsection{Bayesian Coarse Correlated Equilibria}

The notion of {\em Bayesian coarse correlated equilibrium} (see, e.g.,~\cite{forges1993five,bergemann2011correlated}) generalizes to games with incomplete information the standard notion of \emph{coarse correlated equilibrium} for games with complete information by~\cite{moulin1978strategically}.
In particular, we define a Bayesian coarse correlated equilibrium as follows.
\begin{definition}\label{def:bayesian cce}
    A joint strategy $\sigma\in\Sigma$ is a \emph{Bayesian coarse correlated equilibrium} (BCCE) for distributions $(\cF,\cG)$ if, for each player $i\in[n]$, type $v_i\in\V_i$, and deviation bid $b_i'\in\B_i$, it holds:
    \begin{equation}
        \expe_{\substack{\vvec_{-i}\sim\cF_{v_i}\\\gammav\sim\cG}}\Bigg[\sum_{\bv\in\B}\sigma[(v_i,\vvec_{-i}),\bv] \bigg( u_i(\bv,(v_i,\vvec_{-i}),\gammav) \\-  u_i((b'_i,\bv_{-i}),(v_i,\vvec_{-i}),\gammav)\bigg) \Bigg]\ge 0,
    \end{equation}
    where $\cF_{v_i}$ is the posterior distribution over $\bigtimes_{i'\ne i}\V_{i'}$ having observed type $v_i$ for bidder $i$.
\end{definition}

\subsection{Semi-Smoothness}\label{app: semismoothness}

Lucier  and  Paes  Leme~\shortcite{lucier2011gsp} introduce semi-smoothness as an extension of the notion of smoothness by Roughgarden~\shortcite{roughgarden2015intrinsic}. The former notion can be used to derive price of anarchy bounds even in the Bayesian setting with arbitrarily correlated types. 
Formally, by letting $SW(\sigma)$ be the social welfare generated by strategy $\sigma$ for some (implicitly defined) game, the definition of $(\lambda,\mu)$-semi-smoothness reads as follows.

\begin{definition}\label{def:semi smoothness}[Def. 2 by Lucier and Paes Leme~\shortcite{lucier2011gsp}]	
	A game is \emph{$(\lambda,\mu)$-semi-smooth} if there exists some strategy $\sigma'=(\sigma_1',\ldots,\sigma_n')\in\bigtimes_i\Sigma_i$ maximizing the social welfare such that, for any joint strategy profile $\sigma\in\Sigma$ it holds 
	\[
	\sum_{i\in [n]} u_i(\sigma_i',\sigma_{-i})\ge \lambda SW(\sigma')-\mu SW(\sigma),
	\]
	where $SW$ is the social welfare of the mechanism with an arbitrary (fixed) vector of valuations.
\end{definition}

\section{Omitted Proofs}\label{sec:proofs}

\efoneval*
\begin{proof}
To simplify the proof, let $\beta=1$ (the case in which $\beta<1$ follows directly from the case in which $\beta=1$). Let $J_h$ and $J_\ell$ be the set of slots assigned to group $h$ and $\ell$, respectively, by \cref{alg:RR_EF1}. The crucial idea is that \cref{alg:RR_EF1} computes $J_h$ and $J_\ell$ without employing the reported bid profile, which is used only to determine the per-group ranking. In particular, considering the case of group $h$, we know that $\alg_h(J_h,\bv)\ge \alg_h(J_\ell\setminus J_\ell[1],\bv)$. Then, by definition of $\alg$ and by following \cref{alg:RR_EF1} executed with bid profile $\bv$, we have
\[
\sum_{j=1}^{|J_h|-1}\mleft(\alpha_{J_h[j],h}-\alpha_{J_\ell[j+1],h}\mright)b_{I_h[j]} + \alpha_{J_h[|J_h|],h}b_{I_h[|J_h|]}\ge 0.
\]
where we assume $\alpha_{J_{(\cdot)}[j]}=0$ whenever $j\ge |J_{(\cdot)}|$. Since the split in $J_h,J_\ell$ does not depend on the bid profile $\bv$, the above inequality holds for any non negative bid profile $\bv'$, and in particular it holds when $\bv'=\vv$. The case of group $\ell$ is analogous. 
\end{proof}

\EFXgece*
\begin{proof}
The claim of the proof clearly holds for the base case when $J_h=J_\ell=\varnothing$.  For the inductive hypothesis, assume it holds  before the next item $x\in [m]$ is assigned. We distinguish between three cases:
\begin{itemize}
\item Let us first consider the case of a slot $x$ assigned to group $\ell$.  The set of slots assigned to $\ell$ is  $J_\ell \cup \{x\}$. EFX clearly holds for 
group $\ell$, since it is allocated one additional slot. Let us consider group $h$.
We know that, given bid profile $\bv$, for any $y\in J_\ell$,  
\[\alg_h(J_h,\bv) \geq  \beta \alg_h(J_\ell,\bv) \geq \beta \alg_h(J_\ell \cup \{x\}\setminus \{y\},\bv)\]
since $\alg_h(\{y\},\bv) \geq \alg_h(\{x\},\bv)$ (i.e., $y$ has higher click through rate than $x$).
This proves $\beta$-EFX for group $h$. 

\item Next, we consider the case of a swap between the two allocations (Line~\ref{step: gece swap} of Algorithm~\ref{alg:GECE}). Before a swap happens, it holds  $\alg_h(J_h,\bv)< \beta \alg_h(J_\ell,\bv)$ and $\alg_\ell(J_\ell,\bv)< \beta \alg_\ell(J_h,\bv)$,  and, therefore, EFX immediately holds after the swap for any $\beta\leq 1$.

\item Finally, we consider the case of $\alg_h(J_h,\bv)\geq \beta \alg_h(J_\ell,\bv)$ and $\alg_\ell(J_\ell,\bv)\geq \beta \alg_\ell(J_h,\bv)$ before slot $x$ is assigned to group $h$.  In this last case, $\beta$-EFX still holds for group $h$, that receives one more item. It also holds for group $\ell$ since $\alg_\ell(J_\ell,\bv)\geq \beta \alg_\ell(J_h \cup\{x\}\setminus \{y\},\bv)$ for each $y\in J_h$ since $\alg_\ell(\{x\},\bv) \leq \alg_\ell(\{y\},\bv)$. 
\end{itemize}
This proves the statement.
\end{proof}

\kgece*
\begin{proof}
Consider \cref{alg:k-GECE}. We prove the statement by induction. The claim of the proof clearly holds for the base case in which the algorithm is initialized with $J_w\gets\{w\}, w \in [k]$ (\cref{step: gece init}). Moreover, there exists no cycle in the envy graph after the initialization.  As inductive hypothesis, assume the claim holds before the next item $j\in J$ is assigned. We distinguish between two cases:
\begin{itemize}
\item There exists a cycle in the envy graph that needs to be solved before item $j$ is assigned.  Since the allocation was EFX after the previous item $j-1$ was assigned, it is still EFX since all the groups in the cycle get a better allocation (see \cite[Lemma 2.2]{lipton2004approximately}). 
\item Next, consider the case that there is no cycle and item $j$ is assigned to group $g$ that is not envied by any other group. Clearly, the allocation is still EFX for group $g$ since it gets assigned one additional item. For all other groups $w\ne g$, we know that, given bid profile $\bv,$ $\alg_w(J_w,\bv)\geq \alg_w(J_g,\bv)$ before slot $j$ is assigned to group $g$.  In this last case, EFX still holds for group $w$ after item $j$ is assigned since $\alg_w(J_w,\bv)\geq \alg_w(J_g \cup\{j\}\setminus \{y\},\bv)$ for each $y\in J_g$ since $\alg_w(\{j\},\bv) \leq \alg_w(\{y\},\bv)$. 
\end{itemize}
This proves the statement.
\end{proof}

\efficiencyOne*
\begin{proof}
We use the results of Section~\ref{sec:fair gsp auctions} on Group $\beta$-EF1 fairness. 
The social welfare obtained through the composition of the two mechanisms can be lower bounded as follows.

\begin{eqnarray}
\label{eqn:pricevariation}
 \alg^{\textsc{c}}(\bv) &=&  \sum_{i\in I_h\cup I_\ell} \alg_i^{\textsc{c}}(\bv) \nonumber\\
 		      &\geq& \sum_{i\in I_h}\gamma_{g(i)}  b_i \alpha_{\nu^{\textsc{g}}(i),g(i)}   + \sum_{i\in I_\ell} \gamma_{g(i)} b_i   \alpha_{\nu^{\textsc{g}}(i),g(i)}- \sum_{\substack{i\in I_h:\\ \nu^{\textsc{c}}(i)>\nu^{\textsc{g}}(i)}} \gamma_{g(i)} b_{i}  \mleft(\alpha_{\nu^{\textsc{g}}(i),g(i)} - \alpha_{\nu^{\textsc{c}}(i),g(i)}\mright) \label{eq0} \\
 		      &=& \sum_{i\in I_\ell}  \gamma_{g(i)} b_i \alpha_{\nu^{\textsc{g}}(i),g(i)} + \sum_{\substack{i\in I_h:\\ \nu^{\textsc{c}}(i)\leq\nu^{\textsc{g}}(i)}}\gamma_{g(i)}  b_i  \alpha_{\nu^{\textsc{g}}(i),g(i)} + \sum_{\substack{i\in I_h:\\ \nu^{\textsc{c}}(i)>\nu^{\textsc{g}}(i)}} \gamma_{g(i)} b_{i}   \alpha_{\nu^{\textsc{c}}(i),g(i)} \nonumber\\ 
 		      &\geq& 
 		      \alg_\ell^\textsc{g}(\bv) + \sum_{j=1}^{\lceil n/(1+\beta)\rceil} \gamma_{g(\pi^{\textsc{g}}(j))} b_{\pi(j)} \alpha_{\lceil(1+\beta)j\rceil -1,g(\pi^{\textsc{g}}(j)))}\label{eq1}\\
		      &\geq&
		      \alg_\ell^{\textsc{g}} (\bv) + \frac{1}{(1+\beta)} \sum_{j=1}^{n} \gamma_{g(\pi^{\textsc{g}}(j))}  b_{\pi^{\textsc{g}}(j)}  \alpha_{j,g(\pi^{\textsc{g}}(j))}\label{eq2}\\
		      &=& 	
		      \alg_\ell^{\textsc{g}}(\bv) + \frac{1}{(1+\beta)}  \alg^{\textsc{g}}(\bv).\nonumber
\end{eqnarray}

\noindent
Equation~\eqref{eq0} follows from Assumption \ref{assumption1}, which is stating that the application of the composite mechanism will always increase the value of the minority group $\ell$, while the majority group $h$ will experience a value loss.
Equation~\eqref{eq1} follows since the loss of a bidder $\pi^{\textsc{g}}(j)$ is positive if assigned by round robin to a slot with lower quality. We observe that this can only happen  for the first $\lceil n (1/(1+\beta))\rceil $ bidders in the GSP solution  that are shifted from a position  $j$ to a position of lower quality. Given that round robin alternates $\xi_h$ bidders of group $h$ with $\xi_\ell$ bidders of group $\ell$, with $\beta=\xi_\ell/\xi_h$, the landing position of the bidder assigned from GSP at position $j$  is never worse than position $\lceil(1+\beta)j\rceil -1$. Clearly, the bound also holds for the bidders of group $h$ and $\ell$ that are not shifted in the composed mechanism.  
Equation~\eqref{eq2} follows from the fact that \[\sum_{j=1}^{\lceil n/(1+\beta)\rceil} \gamma_{g(\pi^\textsc{g}(j))} b_{\pi^\textsc{g}(j)} \alpha_{\lceil(1+\beta)j\rceil -1,g(\pi^{\textsc{g}}(j)))}\]
has always the term of the first slot followed by at least $\xi_h$ terms out of $\xi_h+\xi_\ell$, with $\xi_\ell\leq \xi_h$, corresponding to bid values that are at least as high as those of GSP. Therefore, at least a \[\frac{\xi_h}{\xi_h+\xi_\ell} = \frac{1}{(1+\beta)}\] fraction of $\alg^{\textsc{g}}(\bv)$ is recovered. This completes the proof.
\end{proof}

\bbOne*
\begin{proof}
We prove that the loss in the value obtained through the $\beta$-Fair GSP mechanisms is at most equal to the payments of GSP. We derive the $2$-budget balance result given that the compensation for each bidder is equal to twice the loss in value according to the pricing rule in~\cref{eq:payments} computed on the bid vector $\bv$. 

Given the GSP allocation $\pi^\textsc{g}$, it can be  observed that the loss in value computed on bid vector $\bv$ is: 
\begin{eqnarray*}
\Delta \alg(\bv) &\leq&  \sum_{j=1}^{\lceil n/(1+\beta)\rceil} b_{\pi^\textsc{g}(j)} \gamma_{g(\pi^\textsc{g}(j))} \mleft(\alpha_{j,g(\pi^\textsc{g}(j))} - \alpha_{\lceil (1+\beta)j\rceil -1,\,\,g(\pi^\textsc{g}(j))}\mright) \\ 
&\leq& \sum_{j=1}^{\lceil n/(1+\beta)\rceil }b_{\pi^\textsc{g}(j+1)}\gamma_{g(\pi^\textsc{g}(j+1))}\alpha_{j+1,\,\,g(\pi^\textsc{g}(j+1))} \\ 
&\leq& \sum_{j=1}^n  p_{\pi^\textsc{g}(j)}^\textsc{g}, 
\end{eqnarray*}

The above derivation is based on the observation that, in the worst case, the round robin mechanism is shifting the first $\lceil n/(1+\beta)\rceil$ bidders to a worse slot of index at most equal to $\lceil(1+\beta)j\rceil -1$. Moreover, the second inequality holds true since, by construction of Algorithm~\ref{alg:RR_EF1} and by~\cref{eq:payments}, the first term of the summation (i.e., for $j=1$) is equal to 0 since by Assumption \ref{assumption2} the first slot is assigned by GSP to the majority group, and since $\gamma_{(\cdot)}\in [0,1]$ for each group $(\cdot)\in\gr$. This concludes the proof.
\end{proof}

\ALGEfx*
\begin{proof}

Let $J_h^{\textsc{g}}$ and $J_\ell^{\textsc{g}}$ be, respectively,  the sets of slots assigned by GSP to groups $h$ and $\ell$, and let $J_h^{\textsc{c}}$ and $J_\ell^{\textsc{c}}$ be the set of slots assigned to groups $h$ and $\ell$ obtained by applying Group Envy-Cycle-Elimination (GECE) with $\beta=1$ to the outcome of GSP. Moreover, let $x_h\in[m]$ and $x_\ell\in[m]$, $x_h\ne x_\ell$, the last slots assigned to $J_h^{\textsc{c}}$ and $J_\ell^{\textsc{c}}$ by GECE, respectively. 
We denote by  $\alg_{(\cdot)}(J)$ the social welfare on bid vector $\bv$ that advertisers of group $(\cdot)\in\gr$ have on items $J$ if assigned in order of decreasing quality.  
By the definition of EFX on bid declarations, it  holds $\alg_h(J_h^{\textsc{c}}) \ge  \alg_h(J_\ell^{\textsc{c}}\setminus \{x_\ell\})$ and $\alg_\ell(J_\ell^{\textsc{c}}) \ge  \alg_\ell(J_h^{\textsc{c}}\setminus \{x_h\})$.

The case of only two slots is trivial. Therefore, we focus on the setting in which there are at least three slots. As a first step, we consider the case in which there are at least two slots assigned to each group. 
The following holds:

\begin{eqnarray*}
\alg_h(J_h^{\textsc{c}}) + \alg_\ell(J_\ell^{\textsc{c}}) &\geq&  
                            \frac{1}{2} (\alg_h(J_h^{\textsc{c}})+ \alg_h(J_\ell^{\textsc{c}}\setminus \{x_\ell\}))+ \frac{1}{2} (\alg_\ell(J_\ell^{\textsc{c}})+ \alg_\ell(J_h^{\textsc{c}}\setminus \{x_h\}))\\
                            &\geq& \frac{1}{2} (\alg_h(J_h^{\textsc{g}}) - \alg_h(\{x_\ell\}))+ \frac{1}{2} (\alg_\ell(J_\ell^{\textsc{g}}) - \alg_\ell(\{x_h\})) \\ 
                            &\geq& \frac{1}{2} (\alg_h(J_h^{\textsc{g}}) + \alg_\ell(J_\ell^{\textsc{g}})) -
                            \frac{1}{2}(\alg_h(J_h^{\textsc{c}}) + \alg_\ell(J_\ell^{\textsc{c}})),
\end{eqnarray*}
where the last inequality follows because, at the time in which the last item is  added to $J_\ell^{\textsc{c}}$, it holds $\alg_h(J_h^{\textsc{c}}) \geq \alg_h(J_\ell^{\textsc{c}})\geq \alg_h(\{x_\ell\})$ in case $J_\ell^{\textsc{c}}$ already contains at least one item when $x_\ell$ is added to $J_\ell^{\textsc{c}}$. Analogously, we can prove that $\alg_\ell(J_\ell^{\textsc{c}}) \geq \alg_\ell(J_h^{\textsc{c}})\geq \alg_\ell(\{x_h\})$.  
This yields the following
\begin{eqnarray*}
\alg_h(J_h^{\textsc{c}}) + \alg_\ell(J_\ell^{\textsc{c}}) &\geq& \frac{1}{3} (\alg_h(J_h^{\textsc{g}}) + \alg_\ell(J_\ell^{\textsc{g}})),
\end{eqnarray*}
for the case in which we have at least two elements in both $J_h^{\textsc{c}}$ and $J_\ell^{\textsc{c}}$. 

We are left to consider the case in which we have one single element either in $J_h^{\textsc{c}}$ or in $J_\ell^{\textsc{c}}$. Assume $J_h^{\textsc{c}} = \{x_h\}$.  If this is the case, this is the first and unique element that is added to $J_h^{\textsc{c}}$. This implies  $\alg_h(J_h^{\textsc{c}}) = \alg_h(\{x_h\})\geq \alg_h(J_\ell^{\textsc{c}})$, because no other elements other than $x_h$ have been added to the allocation for group $h$. Therefore, we obtain:
\begin{eqnarray*}
\alg_h(J_h^{\textsc{c}}) + \alg_\ell(J_\ell^{\textsc{c}}) \hspace{-.2cm} &\geq& \hspace{-.2cm} \alg_h(J_\ell^{\textsc{c}}) + \alg_\ell(J_\ell^{\textsc{c}}) \\
                            &\geq& \hspace{-.2cm}\frac{1}{2} (\alg_h(J_h^{\textsc{c}})+ \alg_h(J_\ell^{\textsc{c}})) +   \frac{1}{2} (\alg_\ell(J_\ell^{\textsc{c}}))\\
                            & \geq& \frac{1}{2} (\alg_h(J_h^{\textsc{g}}))
                            + \frac{1}{2} (\alg_\ell(J_\ell^{\textsc{g}})),
\end{eqnarray*}
where the last derivation follows from Assumption \ref{assumption1} stating that the value for the minority group $\ell$ in the allocation obtained through GSP-EFX is higher than the value for group $\ell$ under the allocation computed through GSP. 

The case of one single element in $J_\ell^{\textsc{c}}$ can be handled similarly. This concludes the proof.
\end{proof}

\bbEfx*
\begin{proof}
We prove that the loss in terms of value between the allocation obtained through GSP-EFX, and the allocation obtained via GSP, is at most twice the payments of GSP. Given that the compensations are equal to twice the loss in value computed on the actual bids, the claim of $4$-budget balance follows.  

Let us denote by $\nu^{\textsc{g}}(i)$ and by $\nu^{\textsc{c}}(i)$ the slots assigned to advertiser $i$ in the allocations computed through GSP and GSP-EFX, respectively. Let us denote by $\texttt{prev}_{g(i)}(i)$ the advertiser that precedes $i$ in the ordering of group $g(i)$.  
We observe that in the GSP-EFX solution provided by the GECE algorithm (Algorithm~\ref{alg:GECE}), the first two slots are always assigned to different groups. If the assignment is in agreement with  the solution of GSP, the claim is proved with a factor of $1$ by bounding the social welfare loss of all the advertisers different from the first bidders of the two groups. In particular, for each bid profile $\bv\in\B$ it holds for the loss of social welfare: 
\begin{eqnarray}
\Delta \alg(\bv) &\leq&  \sum_{i\in J_\ell\cup J_h: \nu^{\textsc{c}}(i)> \nu^{\textsc{g}}(i)} b_{i} \gamma_{g(i)} (\alpha_{\nu^{\textsc{g}}(i),g(i)} - \alpha_{\nu^{\textsc{c}}(i),g(i)}) \nonumber\\ 
&\leq& \sum_{i\in J_\ell\cup J_h: \nu^{\textsc{c}}(i)> \nu^{\textsc{g}}(i)}
b_{i} \gamma_{g(i)} \alpha_{\nu^{\textsc{g}}(i),g(i)}\nonumber \\ 
&\leq& \sum_{i\in J_\ell\cup J_h: \nu^{\textsc{c}}(i)> \nu^{\textsc{g}}(i)}  p_{\texttt{prev}_{g(i)}(i)}^{\textsc{g}}\nonumber \\
&\leq& \sum_{i=1}^n p_i^{\textsc{g}}, \label{eq:delta alg ub}
\end{eqnarray}
with the last two equations derived by the definition of $p^{\textsc{g}}_i$ and $\gamma_{(\cdot)} \in [0,1]$, for each group $(\cdot)\in\gr$.

A similar argument that will be reported in the full version of the paper applies if the GSP-EFX assignment of the first bidder of the two groups is not in agreement with  the solution of GSP. 
\end{proof}

\smooth*
\begin{proof}
Consider an arbitrary bid profile $\bv\in\B$, vector of valuations $\vvec\in\V$, and quality factors $\gammav$. Moreover, the composite mechanism is individually rational at the equilibrium. 
Let $b_i'=v_i/2$ be a deterministic deviation bid for bidder $i$. We have two cases, depending on the outcome of the deviation for bidder $i$. 

\textbf{Case 1:} The composite mechanism \textsc{c} assigns to bidder $i$ a slot $j^\textsc{c}$ such that $\nu^\textsc{g}((b'_i,\bv_{-i}),i)\le j^\textsc{c}$ (i.e., bidder $i$ is penalized in the composition).
Then, for each $v_i$ it holds
\begin{align*}
  u_i^\textsc{c}((b'_i,\bv_{-i}),v_i) & = \alpha_{j^\textsc{c},g(i)} \gamma_{g(i)} v_i - p^\textsc{c}_i\\
    & = \alpha_{j^\textsc{c},g(i)} \gamma_{g(i)} v_i - \mleft(p^\textsc{g}_i-2b_i'\gamma_{g(i)}(\alpha_{\nu^\textsc{g}((b'_i,\bv_{-i}),i),g(i)}-\alpha_{j^\textsc{c},g(i)})\mright)
    \\ & = 
    \alpha_{j^\textsc{c},g(i)} \gamma_{g(i)} v_i + \alpha_{\nu^\textsc{g}((b'_i,\bv_{-i}),i),g(i)}\gamma_{g(i)} v_i - \alpha_{j^\textsc{c},g(i)} \gamma_{g(i)} v_i - p_i^\textsc{g}
    \\ & =
    \alpha_{\nu^\textsc{g}((b'_i,\bv_{-i}),i),g(i)}\gamma_{g(i)} v_i - p_i^\textsc{g}
    \\ & =
    u_i^\textsc{g}((b'_i,\bv_{-i}),v_i),
\end{align*}
where the second equality holds by~\cref{eq:payments}, and the third equality holds by the definition of the deviation bid $b_i'$.

\textbf{Case 2:} The composite mechanism \textsc{c} assigns to bidder $i$ a slot $j^\textsc{c}$ such that $\nu^\textsc{g}((b'_i,\bv_{-i}),i)\ge j^\textsc{c}$ (i.e., bidder $i$ is better of under the composite mechanism). Then $u_i^\textsc{c}((b'_i,\bv_{-i}),v_i)\ge u_i^\textsc{g}((b'_i,\bv_{-i}),v_i)$ for each $v_i$.

Therefore by setting $b_i'=v_i/2$, for any $\bv_{-i}$ and $v_i$, $u_i^\textsc{c}((b'_i,\bv_{-i}),v_i)\ge u_i^\textsc{g}((b'_i,\bv_{-i}),v_i)$.
This also proves the second part of the claim on individual rationality since there exists for each bidder a deviation that does not decrease the utility with respect to GSP that is individually rational, and therefore the composite mechanism is individually rational at the equilibrium.

Let $j^\textsc{g}\defeq\nu^\textsc{g}(\vvec,i)$ be the slot assigned to bidder $i$ through GSP under truthful bidding.
We have to consider two additional cases. 

\textbf{Case 1:} $j^\textsc{g}\ge\nu^\textsc{g}((b'_i,\bv_{-i}),i)$. Then, 
\begin{align*}
u_i^\textsc{g}((b'_i,\bv_{-i}),v_i)
& \ge \frac{1}{2}\alpha_{\nu^\textsc{g}((b'_i,\bv_{-i}),i),g(i)}\gamma_{g(i)} v_i
\\ & \ge \frac{1}{2}\alpha_{j^\textsc{g},g(i)}\gamma_{g(i)} v_i.
\end{align*}

\textbf{Case 2:}
$j^\textsc{g}<\nu^\textsc{g}((b'_i,\bv_{-i}),i)$. Then, since the mechanism is monotone increasing, the effective bid of $\pi^\textsc{g}((b'_i,\bv_{-i}),j^\textsc{g})$ must be greater than or equal to the effective bid of $i$, that is
\begin{align*}
    \alpha_{j^\textsc{g},g(\pi^\textsc{g}((b'_i,\bv_{-i}),j^\textsc{g}))}\gamma_{g(\pi^\textsc{g}((b'_i,\bv_{-i}),j^\textsc{g}))} b_{\pi^\textsc{g}((b'_i,\bv_{-i}),j^\textsc{g})}\ge\frac{1}{2} \alpha_{j^\textsc{g},i}\gamma_{g(i)}v_i.
\end{align*}

Therefore, for any $\bv_{-i}$ and $v_i$,
\begin{eqnarray*}
u_i^\textsc{g}((b'_i,\bv_{-i}),v_i)
\hspace{-.3cm}&\ge&\hspace{-.3cm}
\frac{1}{2} \alpha_{j^\textsc{g},i}\gamma_{g(i)}v_i - \alpha_{j^\textsc{g},g(\pi^\textsc{g}((b'_i,\bv_{-i}),j^\textsc{g}))}\gamma_{g(\pi^\textsc{g}((b'_i,\bv_{-i}),j^\textsc{g}))} b_{\pi^\textsc{g}((b'_i,\bv_{-i}),j^\textsc{g})}.
\end{eqnarray*}

By summing over all players we obtain
\[
\sum_i u_i^\textsc{c}((b'_i,\bv_{-i}),v_i)\ge \sum_i u_i^\textsc{g}((b'_i,\bv_{-i}),v_i)\ge \frac{1}{2}SW^\textsc{g}(\vvec)-\alg^\textsc{g}(\bv).
\]

By~\cref{thm:swl_ref} and by the conservative assumption, we obtain that, for $\beta$-Fair GSP it holds
\begin{align*}
    \sum_i u_i^\textsc{c}((b'_i,\bv_{-i}),v_i) & \ge \frac{1}{2}SW^\textsc{g}(\vvec)-\alg^\textsc{g}(\bv)
    \\ & \ge 
    \frac{1}{2}SW^\textsc{g}(\vvec)-(1+\beta) \alg^\textsc{c}(\bv)
    \\ & \ge 
    \frac{1}{2}SW^\textsc{g}(\vvec)-(1+\beta) SW^\textsc{c}(\bv,\vvec).
\end{align*}

Similarly, by~\cref{thm:rel_GSP_EFX}, in the case of GSP-EFX we obtain that
\begin{align*}
    \sum_i u_i^\textsc{c}((b'_i,\bv_{-i}),v_i) \ge 
    \frac{1}{2}SW^\textsc{g}(\vvec) -3 \cdot SW^\textsc{c}(\bv,\vvec).
\end{align*}
This concludes the proof.
\end{proof}

\pocOne*
\begin{proof}
By~\cref{lemma:semismooth} and \cref{def:smoothness2}, we have that there exists a strategy profile $\sigma'$ of independent bid strategies such that, for any correlated strategy profile $\sigma\in\Sigma$, type profile $\vvec\in\V$, and quality factors $\gammav$, it holds
\begin{equation}
\label{eq:following lemma2}
    \expe_{\bv,\bv'}\mleft[\sum_{i\in[n]} u_i^\textsc{c}((b_i',\bv_{-i}),v_i)\mright]
    \ge \frac{1}{2}\, SW^\textsc{g}(\vvec)-\mleft(1+\beta\mright)\, \expe_{\bv}\mleft[SW^\textsc{c}(\bv,\vvec)\mright],
\end{equation}
\noindent
where we dropped the dependency on $\gammav$ to simplify the notation, and $SW^\textsc{g}(\vvec,\gammav)$ is the social welfare attained by  GSP when agents are truthfully reporting their valuations.
Consider a Bayesian coarse correlated equilibrium $\sigma\in\bcce$ for the composite mechanism. 
Then, by \cref{def:bayesian cce}, we have that for each $i\in I$, and $v_i\in\V_i$,
\[
\expe_{\vvec_{-i},\gammav,\bv}\mleft[ u_i^\textsc{c}(\bv,v_i,\gammav)\mright]\ge\expe_{\vvec_{-i},\gammav,\bv,b_i'}\mleft[ u_i^\textsc{c}((b'_i,\bv_{-i}),v_{i},\gammav)\mright].
\]
Then, 
\begin{align*}
    \expe_{\vvec,\gammav,\bv}\mleft[SW^\textsc{c}(\bv,\vvec,\gammav)\mright] & \ge \expe_{\vvec,\gammav,\bv}\mleft[\sum_{i\in[n]} u_i^\textsc{c}(\bv,v_i,\gammav)\mright]\\
    &\hspace{-0cm} \ge \expe_{\vvec,\gammav,\bv} \mleft[\sum_{i\in[n]}\expe_{b'_i}\mleft[u_i^\textsc{c}((b'_i,\bv_{-i}),v_i,\gammav)\mright]\mright]\\
    &\hspace{-0cm} \ge \expe_{\vvec,\gammav}\mleft[\expe_{\bv,\bv'} \mleft[\sum_{i\in[n]}u_i^\textsc{c}((b'_i,\bv_{-i}),v_i,\gammav)\mright]\mright]\\
    &\hspace{-0cm} \ge \expe_{\vvec,\gammav}\mleft[\frac{1}{2} SW^\textsc{g}(\vvec,\gammav)-\mleft(1+\beta\mright)\, \expe_{\bv}\mleft[SW^\textsc{c}(\bv,\vvec,\gammav)\mright]\mright],
\end{align*}
where the second inequality follows from the definition of Bayesian coarse correlated equilibrium, and the last inequality follows from \cref{lemma:semismooth}.
For any $\vvec$ and $\gammav$ 
we obtain that
\begin{equation*}
\expe_{\vvec,\gammav,\bv}\mleft[SW^\textsc{c}(\bv,\vvec,\gammav)\mright]\ge 
1/2\,\,\expe_{\vvec,\gammav}\mleft[SW^\textsc{g}(\vvec,\gammav)\mright]-\mleft(1+\beta\mright)\, \expe_{\vvec,\gammav,\bv}\mleft[SW^\textsc{c}(\bv,\vvec,\gammav)\mright].
\end{equation*}
This proves our statement.
\end{proof}

\pocTwo*
\begin{proof}
The proof is analogous to that of~\cref{thm:poa beta fgsp}, by employing $\lambda=1/2$, and $\mu=3$.
\end{proof}

\end{document}